\documentclass[11pt]{article}
\usepackage{fullpage}
\usepackage{authblk}
\usepackage{graphicx}
\usepackage{float}
\usepackage{amsmath}
\usepackage{amsfonts}
\usepackage{amssymb}
\usepackage{amsthm}
\usepackage{enumerate}
\usepackage{algorithm}
\usepackage{algpseudocode}
\algrenewcommand{\algorithmicrequire}{\textbf{Input:}}
\algrenewcommand{\algorithmicensure}{\textbf{Output:}}

\usepackage{url}
\usepackage{hyperref}

\newtheorem{theorem}{Theorem}
\newtheorem{lemma}{Lemma}

\newtheorem{proposition}{Proposition}

\newtheorem{assumption}{Assumption}

\newtheorem*{claim*}{Claim}
\newtheorem*{remark*}{Remark}

\newtheorem{definition}{Definition}
\newtheorem*{definition*}{Definition}

\makeatletter

\renewcommand\section{%
  \@startsection{section}{1}
                {\z@}%
                {-3.5ex \@plus -1ex \@minus -.2ex}%
                {2.3ex \@plus.2ex}%
                {\large\bfseries}
}
\renewcommand\subsection{%
  \@startsection{subsection}{2}
                {\z@}%
                {-3.25ex\@plus -1ex \@minus -.2ex}%
                {1sp}
                {\normalsize\bfseries}
}
\renewcommand\subsubsection{%
  \@startsection{subsubsection}{3}
                {\z@}%
                {-3.25ex\@plus -1ex \@minus -.2ex}%
                {1sp}
                {\normalfont\normalsize}
}

\makeatother

\title{{\Large\bf  A Complete Characterization of Convexity in Flow Games}\thanks{Supported in part by the National Natural Science Foundation of China (Nos.\,12671377 and 12171444) and the Shandong Provincial Natural Science Foundation (No.\, ZR2025MS104).
}}

\author{Han Xiao\thanks{Corresponding author. Email: hxiao@ouc.edu.cn.}, Luying Zhang, Qizhi Fang}

\affil{School of Mathematical Sciences\\Ocean University of China\\Qingdao, China}

\date{\today}

\begin{document}


\maketitle

\openup 1.2\jot


\begin{abstract}

Flow games coincide precisely with non-negative totally balanced games.
We give a complete characterization of convexity for flow games with one arc per player.
We prove that such a game is convex if and only if all simple source-to-sink paths can simultaneously carry their full path capacities.
Equivalently, after deleting null arcs, the topology is a funnel network, and every arc can meet the combined demand of the paths through it.
We provide polyhedral characterizations of convexity.
In primal form, the game is convex if and only if the path-flow polytope is a box, allowing each path flow to vary independently up to its capacity.
In dual form, convexity is equivalent to the existence of a minimum source-to-sink cut consisting entirely of private arcs.
An arc is private if it belongs to exactly one simple source-to-sink path.
We also give a polynomial-time algorithm that tests whether a given flow game is convex.

\hfill



\hfill

\end{abstract}

\newpage

\section{Introduction}

A cooperative game is convex if each player's marginal contribution is non-decreasing as the coalition grows.
Convexity gives both allocation guarantees and computational advantages.
The core is the convex hull of the marginal contribution vectors~\cite{Shap71}, and their average, the Shapley value, belongs to the core.
Convex games also admit \emph{population monotonic allocation schemes} (PMAS)~\cite{Spru90}, in which each player's payoff is non-decreasing as the coalition expands.
The \emph{nucleolus} coincides with the \emph{kernel}~\cite{MPS72, MPS79} and can be computed in polynomial time given a polynomial-time value oracle~\cite{FKK01}.

An important class of cooperative games is that of \emph{combinatorial optimization games}~\cite{DIN99}, where the characteristic function is defined by the optimal value of an underlying combinatorial optimization problem. 
A prominent and widely studied model in this category is the \emph{flow game}~\cite{KZ82a,KZ82b}, derived from the maximum flow problem on directed networks.\footnote{The term \emph{flow game} is also used for non-cooperative models of strategic routing and flow allocation~\cite{BHS25,CCCW21,GKV19}.}
Flow games provide a canonical framework for resolving profit sharing and resource distribution in capacitated transportation, supply chain, and telecommunication networks.

Kalai and Zemel~\cite{KZ82b} proved that the flow games in which a player may own several arcs coincide with the class of non-negative \emph{totally balanced games}.
The one-arc-per-player games studied here form a special case, so every such game and each of its subgames have a nonempty core.
Core nonemptiness alone does not guarantee computational tractability.
In fact, testing core membership for general flow games was proven to be co-\textsf{NP}-complete~\cite{FZCD02}. 
Although efficient algorithms exist for specific subclasses such as simple flow games, whose arcs have unit capacities~\cite{DFS06, PRB06, RMPT96}, computing the nucleolus remains \textsf{NP}-hard in general~\cite{DFS06, KP09}.

These computational difficulties motivate the search for network conditions that characterize convexity.
Previous studies connect network structure with flow sensitivity and cooperative solution concepts.
Granot and Veinott~\cite{GV85} studied how arc orientations determine complementarity and substitutability in parametric network-flow problems.
Granot and Granot~\cite{GG92} studied games induced by network optimization problems on \emph{directed augmented trees}.
They proved that the nucleolus coincides with the kernel for this class.
Here we study which combinations of topology and capacities yield convexity in one-arc-per-player flow games.

We prove that such a game is convex if and only if all simple source-to-sink paths can simultaneously carry their full path capacities.
Each path's capacity is the minimum capacity of its arcs.
For strictly positive capacities, deleting arcs that belong to no source-to-sink path allows us to separate the characterization into topology and capacity conditions.
The topological requirement is acyclicity together with the existence, on every source-to-sink path, of an arc belonging to no other such path.
We call such a topology a \emph{funnel network}, or simply a \emph{funnel}.
An arc is \emph{private} if it belongs to exactly one path, and \emph{shared} otherwise.
Each shared arc must have capacity at least the sum of the capacities of the paths containing it.
These topology and capacity conditions are also equivalent to attaining the sum of all path capacities as the maximum-flow value.
They are further equivalent to a decomposition into path-unanimity games weighted by path capacities.
The structural characterization yields a polynomial-time algorithm that recognizes convexity for arbitrary input networks.

For networks without null arcs, we also give primal and dual characterizations of convexity.
The funnel condition is equivalent to the existence of a unique path decomposition for every flow, and the path-flow polytope is a box if and only if the game is convex.
On the primal side, the box structure then extends to every subnetwork, and the maximum flow is unique.
On the dual side, the game is convex if and only if a minimum source-to-sink cut consists entirely of private arcs.
For convex instances, we also describe all optimal dual solutions: each path-covering constraint is tight, every dual price vanishes on arcs that the maximum flow leaves unsaturated, and these solutions can be chosen monotonically across coalitions.

The remainder of this paper is organized as follows. 
Section~\ref{sec:preliminaries} provides the necessary preliminaries on cooperative games and flow networks. 
Section~\ref{sec:convexity_characterization} gives structural, primal, and dual characterizations of convexity, a polynomial-time recognition algorithm, and allocation formulas.
Finally, Section~\ref{sec:conclusion} offers concluding remarks.

\section{Preliminaries}
\label{sec:preliminaries}

This section reviews the fundamental concepts from cooperative game theory and network flow theory that form the basis of our analysis.

\subsection{Cooperative Game Theory}

A \emph{cooperative game} in characteristic function form is a pair $(N,\gamma)$, often simply denoted by $\Gamma$, where $N$ is a finite set of \emph{players} and $\gamma: 2^N \to \mathbb{R}$ is the \emph{characteristic function}.
The set $N$ is called the \emph{grand coalition} and any subset $S\subseteq N$ is called a \emph{coalition}.
The function $\gamma$ assigns a value to each coalition, with $\gamma(\emptyset)=0$.
A player $i \in N$ is called a \emph{null player} if $\gamma(S \cup \{i\}) = \gamma(S)$ for all $S \subseteq N \setminus \{i\}$.
A player $i \in N$ is called an \emph{essential player} if $\gamma(N) - \gamma(N \setminus \{i\}) > 0$.

A game $\Gamma=(N,\gamma)$ is \emph{convex} if its characteristic function $\gamma$ is \emph{supermodular}, meaning that the marginal contribution of any player is non-decreasing as the coalition grows:
\begin{equation*}
\gamma(S \cup \{i\}) - \gamma(S) \leq \gamma(T \cup \{i\}) - \gamma(T)
\end{equation*}
for any player $i \in N$ and any two coalitions $S \subseteq T \subseteq N \setminus \{i\}$. 
A fundamental class of convex games is the family of unanimity games. 
For any nonempty coalition $T\subseteq N$, the \emph{unanimity game} is the pair $(N,\zeta_T)$, where the characteristic function $\zeta_T$ is defined as:
\begin{equation*}
\zeta_T(S) =  
\begin{cases}
1 & \text{if } T \subseteq S, \\
0 & \text{otherwise.}
\end{cases}
\end{equation*}
The set $\{\zeta_T\}_{\emptyset \neq T\subseteq N}$ forms a basis for the vector space of characteristic functions.
Thus, the characteristic function $\gamma$ can be uniquely expressed as the linear combination
\begin{equation*}
\gamma =\sum_{\emptyset\neq T\subseteq N} \Delta_{\gamma}(T) \zeta_T,
\end{equation*}
where the coefficients $\Delta_{\gamma}(T)$ are known as the \emph{Harsanyi dividends}.
A sufficient condition for a game to be convex is that all its Harsanyi dividends are non-negative~\cite{Hars59}.

A primary objective in cooperative game theory is to formulate principles for distributing the total value $\gamma(N)$ among the players. 
An \emph{allocation} is a vector $\boldsymbol{x} \in \mathbb{R}^{|N|}$. 
For any coalition $S \subseteq N$, we denote $x(S) = \sum_{i \in S} x_i$. 
An allocation is \emph{efficient} if $x(N) = \gamma(N)$.
The \emph{core} of a game $\Gamma$, denoted by $\mathcal{C}(\Gamma)$, is the set of all efficient allocations that satisfy \emph{coalitional rationality}, meaning no coalition has an incentive to secede from the grand coalition. 
Formally:
\begin{equation*}
\mathcal{C}(\Gamma) = \left\{ \boldsymbol{x} \in \mathbb{R}^{|N|} \;\middle|\; x(N) = \gamma(N) \text{ and } x(S) \ge \gamma(S) \text{ for all } S \subseteq N \right\}.
\end{equation*}
The \emph{Shapley value}~\cite{Shap53} averages the marginal contribution vectors over all player orders. For each player $i\in N$,
\begin{equation}\label{eq:shapley_value_definition}
\Phi_i(\gamma) = \sum_{S \subseteq N \setminus \{i\}} \frac{|S|! (|N|-|S|-1)!}{|N|!} \left[ \gamma(S \cup \{i\}) - \gamma(S) \right].
\end{equation}
For a convex game, each such vector, and hence its average, belongs to the core.

A game $\Gamma$ is called \emph{balanced} if its core is nonempty.
For any nonempty coalition $S \subseteq N$, the \emph{subgame} restricted to $S$ is the pair $(S, \gamma|_S)$, where $\gamma|_S$ is the restriction of the characteristic function to subsets of $S$.
A game $\Gamma$ is called \emph{totally balanced} if the core of every subgame is nonempty.

A \emph{population monotonic allocation scheme} (PMAS)~\cite{Spru90} for $\Gamma=(N,\gamma)$ is a collection of vectors $\{\boldsymbol{x}_S\}_{\emptyset \neq S\subseteq N}$, where $\boldsymbol{x}_S = (x_{S,i})_{i \in S} \in \mathbb{R}^{|S|}$.
Each allocation is efficient: $x_S(S) = \gamma(S)$ for every nonempty coalition $S \subseteq N$.
The allocations are population monotonic: $x_{S,i} \le x_{T,i}$ whenever $\emptyset \neq S \subseteq T \subseteq N$ and $i \in S$.

If a game $\Gamma$ admits a PMAS, then for each nonempty coalition $S \subseteq N$, the specified vector $\boldsymbol{x}_S$ belongs to the core of the subgame $(S, \gamma|_S)$.
Thus, any game possessing a PMAS is totally balanced.
Every convex game admits a PMAS~\cite{Spru90}, and hence is totally balanced~\cite{Shap71}.
The resulting inclusions are
\[ \text{Convex Games} \subseteq \text{Games with a PMAS} \subseteq \text{Totally Balanced Games} \subseteq \text{Balanced Games}. \]

\subsection{Flow Games}

We consider flow games with one arc per player and no public arcs.
Let $D=(V,E;c;s,t)$ be a \emph{flow network} consisting of a vertex set $V$, an arc set $E$, distinct terminals $s,t \in V$, and an arc capacity function $c: E \to \mathbb{R}_{> 0}$.
We write $G_D=(V,E;s,t)$ for the underlying directed topology of $D$, retaining the vertices, arcs, and terminals but omitting the capacity function.
We use $G$ for a directed topology considered independently of a particular flow network.
The associated \emph{flow game} is the pair $\Gamma_D=(E,\gamma)$, with each arc constituting a distinct player.
When the network needs to be specified, we write $\gamma_D$ for its characteristic function.
Players cooperate with each other to maximize flows from $s$ to $t$.
For any coalition $S \subseteq E$, the value of $\gamma(S)$ is given by the value of a maximum $s$-$t$ flow in the subnetwork $D_S = (V, S; c|_S; s, t)$, where only arcs in $S$ are available.

Our analysis uses the \emph{flow decomposition theorem}~\cite{Schr03}: every maximum $s$-$t$ flow can be decomposed into a sum of flows along $s$-$t$ paths and flows along cycles.
Since cycles do not contribute to the net flow value, the study of the flow game can be equivalently framed as optimizing flows over the set of paths. 
Accordingly, we define the basic structural concepts as follows. 
A \emph{path} is a sequence of distinct vertices and arcs connecting a starting vertex to an ending vertex. 
For convenience, we identify each path with its set of arcs.
A \emph{cycle} is a sequence of vertices and arcs where the initial and terminal vertices coincide, while all internal vertices are distinct. 
Unless otherwise stated, all paths and cycles are assumed to be directed.
A network is \emph{acyclic} if it contains no cycles. 

Let $\mathcal{P}$ denote the set of all $s$-$t$ paths in $D$. 
For any $S \subseteq E$, let $\mathcal{P}_S = \{ P \in \mathcal{P} \mid P \subseteq S \}$ denote the set of all $s$-$t$ paths available in the subnetwork $D_S$. 
For any arc $e \in E$, let $\mathcal{P}_e = \{ P \in \mathcal{P} \mid e \in P \}$ be the set of all $s$-$t$ paths containing arc $e$.
The maximum-flow value $\gamma(S)$ is the optimal value of the following linear programming problem:
\begin{equation}\label{eq:path_flow_primal}
\begin{aligned}
\max \quad & \sum_{P \in \mathcal{P}_S} z_P \\
\text{s.t.} \quad & \sum_{P \in \mathcal{P}_S : e \in P} z_P \le c(e), \quad \forall e \in S, \\
& z_P \ge 0, \quad \forall P \in \mathcal{P}_S.
\end{aligned}
\end{equation}
The dual of \eqref{eq:path_flow_primal} assigns a non-negative weight $y_e$ to each arc $e \in S$, reflecting its dual price:
\begin{equation}\label{eq:path_flow_dual}
\begin{aligned}
\min \quad & \sum_{e \in S} c(e) y_e \\
\text{s.t.} \quad & \sum_{e \in P} y_e \ge 1, \quad \forall P \in \mathcal{P}_S, \\
& y_e \ge 0, \quad \forall e \in S.
\end{aligned}
\end{equation}
By the strong duality theorem, the primal and dual optimal values coincide. 
The dual variables $y_e$ represent the marginal importance of arc capacities to the total flow value.
By the max-flow min-cut theorem, $\gamma(S)$ is equivalently the minimum capacity of an $s$-$t$ cut in $D_S$. 
Formally, an \emph{$s$-$t$ cut} in $D$ is a subset of arcs $K \subseteq E$ such that there is no $s$-$t$ path in $D_{E \setminus K}$.
Its capacity is $\sum_{e \in K} c(e)$.
Given a flow $f$, an arc is \emph{saturated} if $f(e)=c(e)$, and \emph{unsaturated} otherwise.

For any directed path $P$, its capacity is defined as $c(P)=\min_{e\in P}c(e)$.
An arc $e \in E$ is a \emph{bottleneck arc} if there exists a path $P \in \mathcal{P}$ such that $e \in P$ and $c(e) = c(P)$.
Let $B$ denote the set of all bottleneck arcs in $D$.  
Given a path $P$ and two vertices $u, v$ on $P$, we denote by $P[u,v]$ the contiguous subpath from $u$ to $v$ along $P$. 

For the structural analysis, we use the following normalization.

\begin{assumption}\label{ass:effective_network}
Every arc belongs to at least one simple $s$-$t$ path.
\end{assumption}

For any flow network $D$, let $\widehat{E}=\bigcup_{P\in\mathcal{P}}P$ and $\widehat{D}=(V,\widehat{E};c|_{\widehat{E}};s,t)$.
The next lemma shows that this normalization preserves all coalition values after deleting null players.

\begin{lemma}\label{lem:null_arc_reduction}
The network $\widehat{D}$ satisfies Assumption~\ref{ass:effective_network}, and
\begin{equation}\label{eq:null_arc_reduction}
\gamma_D(S)=\gamma_{\widehat{D}}(S\cap \widehat{E})\qquad\text{for every }S\subseteq E.
\end{equation}
In particular, the arcs in $E\setminus \widehat{E}$, called \emph{null arcs}, are exactly the null players, and $\Gamma_D$ is convex if and only if $\Gamma_{\widehat{D}}$ is convex.
\end{lemma}

\begin{proof}
We compare coalition values through the path-flow formulation~\eqref{eq:path_flow_primal}.
By definition, every arc of $\widehat{E}$ lies on an $s$-$t$ path contained in $\widehat{E}$, so $\widehat{D}$ satisfies Assumption~\ref{ass:effective_network}.
Fix a coalition $S\subseteq E$.
The subnetworks $D_S$ and $\widehat{D}_{S\cap\widehat{E}}$ have exactly the same $s$-$t$ paths.
The capacity constraints for arcs in $S\setminus \widehat{E}$ have zero left-hand sides, while all remaining constraints are unchanged.
The two path-flow problems therefore have the same optimal value, proving~\eqref{eq:null_arc_reduction}.

It follows that every arc outside $\widehat{E}$ is a null player.
If $e\in \widehat{E}$, choose a path $P\in\mathcal{P}_e$.
Then
\[
\gamma_D(P)-\gamma_D(P\setminus\{e\})=c(P)>0,
\]
so $e$ is not a null player.

Finally, if $\Gamma_D$ is convex, then so is its restriction $\Gamma_{\widehat{D}}$ to $\widehat{E}$.
Conversely, suppose that $\Gamma_{\widehat{D}}$ is convex.
For any $S,T\subseteq E$,~\eqref{eq:null_arc_reduction} and supermodularity give
\[
\begin{aligned}
\gamma_D(S)+\gamma_D(T)
&=\gamma_{\widehat{D}}(S\cap \widehat{E})+\gamma_{\widehat{D}}(T\cap \widehat{E})\\
&\le\gamma_{\widehat{D}}((S\cup T)\cap \widehat{E})+\gamma_{\widehat{D}}((S\cap T)\cap \widehat{E})\\
&=\gamma_D(S\cup T)+\gamma_D(S\cap T).
\end{aligned}
\]
Hence $\Gamma_D$ is convex.
\end{proof}

Throughout the network analysis below, $D$ denotes a network satisfying Assumption~\ref{ass:effective_network}.
Algorithm~\ref{alg:convexity_recognition} instead accepts arbitrary input networks, with no prior deletion of null arcs.

\section{Characterizing Convex Flow Games}
\label{sec:convexity_characterization}

We first derive necessary network conditions from convexity.
Theorem~\ref{thm:convexity_characterization} gives a structural characterization of convexity in terms of network topology and capacity inequalities.
We then give primal and dual characterizations through the path-flow box and optimal dual solutions supported on private arcs.
The structural conditions also yield a convexity recognition algorithm for arbitrary networks.
Finally, we derive allocation formulas from the path decomposition.

\subsection{Necessary Properties}
\label{subsec:necessary_properties}

\begin{proposition}[Arc essentiality]\label{prop:arc_essentiality}
If $\Gamma_D$ is convex, then every arc $e \in E$ is essential and carries strictly positive flow in every maximum flow of the network $D$.
\end{proposition}

\begin{proof}
We compare an arc's marginal contribution in a path coalition with its contribution to the grand coalition.
Fix $e\in E$ and choose an $s$-$t$ path $P$ containing $e$, as guaranteed by Assumption~\ref{ass:effective_network}.
Convexity gives
\[
\gamma(E)-\gamma(E\setminus\{e\})
\ge \gamma(P)-\gamma(P\setminus\{e\})
=c(P)>0.
\]
Thus $e$ is essential.
A maximum flow assigning zero to $e$ would remain feasible after deleting $e$, contradicting this strict inequality.
Hence every maximum flow assigns strictly positive flow to every arc.
\end{proof}

\begin{proposition}[Acyclicity]\label{prop:acyclicity}
If $\Gamma_D$ is convex, then the network $D$ must be acyclic.
\end{proposition}

\begin{proof}
Assume to the contrary that the network $D$ contains a cycle $C$.
Let $f$ be an arbitrary maximum flow in $D$. 
Let $\delta=\min_{e \in C} f(e)$.
By Proposition~\ref{prop:arc_essentiality}, the convexity of $\Gamma_D$ implies that $f(e) > 0$ for every arc $e \in C$.
Hence we have $\delta>0$.
Let $f'$ be a new flow constructed by subtracting $\delta$ units of flow along the cycle $C$:
\[
f'(e) = 
\begin{cases} 
f(e) - \delta & \text{if } e \in C, \\
f(e) & \text{if } e \notin C.
\end{cases}
\]
Since $C$ is a cycle, flow conservation is maintained at every vertex, and the net flow value from $s$ to $t$ remains unchanged. Thus, $f'$ is also a maximum flow.

However, by the choice of $\delta$, there exists an arc $e^* \in C$ such that $f(e^*) = \delta$ and consequently, $f'(e^*) = 0$.
This contradicts Proposition~\ref{prop:arc_essentiality}, which requires every arc to carry strictly positive flow in every maximum flow. 
Therefore, no cycle exists in $D$.
\end{proof}

\begin{proposition}[Subpath disjointness]\label{prop:subpath_disjointness}
Let $P$ and $Q$ be two distinct $s$-$t$ paths.
Let $P[s,u]=Q[s,u]$ be their longest common prefix, and let $P[v,t]=Q[v,t]$ be their longest common suffix.
If $\Gamma_D$ is convex, then $P[u,v]$ and $Q[u,v]$ are internally vertex-disjoint.
Moreover, the joint capacity of $P[u,v]$ and $Q[u,v]$ is bounded by
    \[ c(P[u,v]) + c(Q[u,v]) \leq \min \{ c(P[s,u]), c(P[v,t]) \}. \]
We take $c(P[s,u])=+\infty$ for $u=s$ and $c(P[v,t])=+\infty$ for $v=t$ to handle trivial subpaths.
\end{proposition}

\begin{proof}
We first rule out an internal intersection and then prove the capacity bound.
By Proposition~\ref{prop:acyclicity}, $D$ is acyclic.
Hence the common vertices of $P$ and $Q$ occur in the same order along both paths.
Assume to the contrary that the subpaths $P[u,v]$ and $Q[u,v]$ share an internal vertex.
Let $u'$ be their first common interior vertex after $u$, and $v'$ be their last common interior vertex before $v$.
Note that $u'$ and $v'$ might coincide.
By acyclicity and the choices of these vertices, neither $Q[u,u']$ nor $Q[v',v]$ has an internal vertex on $P$.
Both subpaths are arc-disjoint from $P$, since the paths leave $u$ through different arcs and enter $v$ through different arcs.
Thus $P[u,u']$ and $Q[u,u']$ form one pair of alternative subpaths, while $P[v',v]$ and $Q[v',v]$ form another.

Replacing the subpath $P[v',v]$ of $P$ by $Q[v',v]$ yields another simple $s$-$t$ path
\[
\tilde P=P[s,v']\cup Q[v',v]\cup P[v,t].
\]
Applying supermodularity to $P$ and $\tilde P$ gives
\[
\gamma(P\cup\tilde P)-\gamma(P)\ge\gamma(\tilde P)-\gamma(P\cap\tilde P).
\]
The coalition $P\cap\tilde P=P[s,v']\cup P[v,t]$ contains no $s$-$t$ path, so $\gamma(P\cap\tilde P)=0$.
Expressing the three terms of the inequality with capacities gives
\begin{equation}\label{eq:subpath_intersection_inequality}
\begin{aligned} 
&\min\{c (P[s,v']), c (P[v',v])+c (Q[v',v]), c (P[v,t])\}\\
\geq &\min\{c (P[s,v']), c (P[v',v]), c (P[v,t])\} + \min\{c (P[s,v']), c (Q[v',v]), c (P[v,t])\}.
\end{aligned}
\end{equation}
Since $\min\{c (P[s,v']), c (P[v,t])\}\geq \min\{c (P[s,v']), c (P[v',v])+c (Q[v',v]), c (P[v,t])\}$, 
it follows from inequality \eqref{eq:subpath_intersection_inequality} that $\min\{c(P[s,v']), c(P[v,t])\}> \min\{c(P[s,v']), c(P[v',v]), c(P[v,t])\}$, implying $\min\{c(P[s,v']), c(P[v,t])\} > c(P[v',v])$.
Note that $P[u, u'] \subseteq P[s, v']$.
It follows that $c(P[u, u']) \geq c(P[s, v']) > c(P[v', v])$.
On the other hand, consider the path $P$ and the subpath $Q[u, u']$.
With a similar argument, we can show that $c(P[v', v]) \geq c(P[u', t]) > c(P[u, u'])$.
A contradiction occurs.
Thus, $P[u,v]$ and $Q[u,v]$ must be internally vertex-disjoint.

Next, we prove the capacity condition.
Replacing the subpath $P[u,v]$ of $P$ by $Q[u,v]$ yields the simple $s$-$t$ path
\[
\hat P=P[s,u]\cup Q[u,v]\cup P[v,t].
\]
The internal disjointness just proved and the choices of $u,v$ show that $Q[u,v]$ is arc-disjoint from $P$.
Applying supermodularity to $P$ and $\hat P$ gives
\[
\gamma(P\cup\hat P)-\gamma(P)\ge\gamma(\hat P)-\gamma(P\cap\hat P).
\]
Here $P\cap\hat P=P[s,u]\cup P[v,t]$, so $\gamma(P\cap\hat P)=0$.
As before, expressing the three terms of the inequality with capacities gives
\begin{equation}\label{eq:subpath_capacity_inequality}
\begin{aligned} 
&\min\{c (P[s,u]), c (P[u,v])+c (Q[u,v]), c (P[v,t])\}\\
\geq &\min\{c(P[s,u]), c(P[u,v]), c(P[v,t])\} + \min\{c(P[s,u]), c(Q[u,v]), c(P[v,t])\}.
\end{aligned}
\end{equation}
Since $\min\{c (P[s,u]), c (P[v,t])\}\geq \min\{c (P[s,u]), c (P[u,v])+c (Q[u,v]), c (P[v,t])\}$, 
the inequality \eqref{eq:subpath_capacity_inequality} implies that 
\[\min\{c (P[s,u]), c (P[v,t])\} > \min\{c (P[s,u]), c (P[u,v]), c (P[v,t])\}\]
and 
\[\min\{c (P[s,u]), c (P[v,t])\} > \min\{c (P[s,u]), c (Q[u,v]), c (P[v,t])\}.\]
It follows that $\min\{c (P[s,u]), c (P[v,t])\}> c(P[u,v])$ and $\min\{c (P[s,u]), c (P[v,t])\}> c(Q[u,v])$.
By using the inequality \eqref{eq:subpath_capacity_inequality} again, we have
\begin{equation*}
\begin{aligned}
&\min\{c (P[s,u]), c (P[v,t])\}\\
\geq&\min\{c (P[s,u]), c (P[u,v])+c (Q[u,v]), c (P[v,t])\}\\
\geq &\min\{c (P[s,u]), c (P[u,v]), c (P[v,t])\} + \min\{c (P[s,u]), c (Q[u,v]), c (P[v,t])\}\\
= & c(P[u,v])+ c(Q[u,v]).
\end{aligned}
\end{equation*}
\end{proof}

Proposition~\ref{prop:subpath_disjointness} constrains alternative subpaths with a common prefix and suffix.
Their combined capacity cannot exceed the capacity of either shared part.
We use these properties to show that any arc attaining a path's minimum capacity must belong to that path alone.

\begin{definition}[Bottleneck exclusivity]\label{def:bottleneck_exclusivity}
A network $D$ satisfies \emph{bottleneck exclusivity} if every bottleneck arc $b \in B$ belongs to exactly one $s$-$t$ path $P \in \mathcal{P}$.
\end{definition}

\begin{proposition}[Bottleneck exclusivity]\label{prop:bottleneck_exclusivity}
If $\Gamma_D$ is convex, then the network $D$ must satisfy bottleneck exclusivity.
\end{proposition}

\begin{proof}
We use the subpath structure and capacity bound in Proposition~\ref{prop:subpath_disjointness} to rule out a shared bottleneck.
Let $b\in B$, and choose a path $P$ containing $b$ with $c(b)=c(P)$.
Suppose that a distinct $s$-$t$ path $Q$ also contains $b$.
Let $P[s,u]=Q[s,u]$ be their longest common prefix, and let $P[v,t]=Q[v,t]$ be their longest common suffix.
By Proposition~\ref{prop:subpath_disjointness}, the subpaths $P[u,v]$ and $Q[u,v]$ are internally vertex-disjoint.
Their first arcs differ, so they are also arc-disjoint.
Consequently, the common arc $b$ lies in $P[s,u]$ or $P[v,t]$.
The capacity bound in the same proposition therefore gives
\[
c(P[u,v])+c(Q[u,v])
\le\min\{c(P[s,u]),c(P[v,t])\}
\le c(b).
\]
Since $Q[u,v]$ is nonempty and all arc capacities are strictly positive, $c(Q[u,v])>0$.
It follows that
\[
c(b)=c(P)\le c(P[u,v])<c(b),
\]
a contradiction.
Thus every bottleneck arc belongs to exactly one $s$-$t$ path.
\end{proof}

\begin{definition}[Capacity sufficiency]\label{def:capacity_sufficiency}
A network $D$ satisfies \emph{capacity sufficiency} if for every $e\in E$,
\[
c(e)\ge\sum_{P\in\mathcal{P}_e}c(P).
\]
\end{definition}

Thus capacity sufficiency means that all $s$-$t$ paths can simultaneously carry their full path capacities.

\begin{proposition}[Capacity sufficiency]\label{prop:capacity_sufficiency}
If $\Gamma_D$ is convex, then $D$ satisfies capacity sufficiency.
\end{proposition}

\begin{proof}
We use supermodularity and bottleneck exclusivity to show that all paths can simultaneously carry their full capacities.
List the paths as $\mathcal{P}=\{P_1,\ldots,P_r\}$ and choose a bottleneck $b_i$ of each $P_i$.
By Proposition~\ref{prop:bottleneck_exclusivity}, $b_i$ belongs to no other path.
Set $U_0=\emptyset$ and $U_i=\bigcup_{j=1}^i P_j$.
The intersection $P_i\cap U_{i-1}$ omits $b_i$ and is a proper subset of the simple path $P_i$, so its value is zero.
Supermodularity implies
\[
\gamma(U_i)\ge\gamma(U_{i-1})+\gamma(P_i)-\gamma(P_i\cap U_{i-1})
=\gamma(U_{i-1})+c(P_i).
\]
Since $U_r=E$ by Assumption~\ref{ass:effective_network}, induction gives $\gamma(E)\ge\sum_{P\in\mathcal{P}}c(P)$.
On the other hand, a maximum flow admits a path decomposition with coefficients $z_P\ge0$ satisfying $z_P\le c(P)$.
Therefore
\[
\gamma(E)=\sum_{P\in\mathcal{P}}z_P\le\sum_{P\in\mathcal{P}}c(P).
\]
Equality holds, forcing $z_P=c(P)$ for every path.
Feasibility of this decomposition yields for any $e\in E$,
\[
\sum_{P\in\mathcal{P}_e}c(P)\le c(e),
\]
which proves capacity sufficiency.
\end{proof}

Together, bottleneck exclusivity and capacity sufficiency give a minimum $s$-$t$ cut with one bottleneck from each path.
For each path $P\in\mathcal{P}$, choose an arbitrary bottleneck $b_P$, called its \emph{representative bottleneck}.
Write $B^*=\{b_P:P\in\mathcal{P}\}$ for the resulting set of representatives.

\begin{proposition}[Bottleneck cuts]\label{prop:bottleneck_cuts}
If $\Gamma_D$ is convex, then any choice of representative bottlenecks yields a minimum $s$-$t$ cut $B^*$ meeting every path in $\mathcal{P}$ exactly once.
Moreover,
\[
\sum_{b\in B^*}c(b)=\sum_{P\in\mathcal{P}}c(P)=\gamma(E).
\]
\end{proposition}

\begin{proof}
By Proposition~\ref{prop:bottleneck_exclusivity}, the representatives are distinct and $B^*\cap P=\{b_P\}$ for every path $P$.
Thus $B^*$ meets every $s$-$t$ path exactly once and is a cut of capacity $\sum_{P\in\mathcal{P}}c(P)$.
By Proposition~\ref{prop:capacity_sufficiency}, sending $c(P)$ units along every path $P$ is feasible.
The resulting flow has the same value as the capacity of $B^*$, so the flow is maximum and the cut is minimum.
\end{proof}

\subsection{Structural Characterizations}
\label{subsec:structural_characterization}

We express the necessary network conditions by separating topology from capacities, and then prove that they are sufficient.
The underlying topology of a convex flow game is necessarily a funnel network, a structure we define next.

\begin{definition}[Funnel network]\label{def:funnel_network}
Let $G=(V,E;s,t)$ be a directed network topology satisfying Assumption~\ref{ass:effective_network}.
An arc $e$ is \emph{private} if $|\mathcal{P}_e|=1$, and \emph{shared} if $|\mathcal{P}_e|\ge2$.
We call $G$ a \emph{funnel network}, or simply a \emph{funnel}, if it is acyclic and every $P\in\mathcal{P}$ contains a private arc.
\end{definition}

For each vertex $v$ of a directed topology $G=(V,E;s,t)$, let $\pi^{-}(v)$ and $\pi^{+}(v)$ denote the numbers of $s$-$v$ paths and $v$-$t$ paths, respectively.
Paths of length zero are counted.
The following characterization describes a funnel in terms of these path counts.

\begin{lemma}\label{lem:funnel_path_counts}
Let $G=(V,E;s,t)$ be an acyclic topology satisfying Assumption~\ref{ass:effective_network}.
Then $G$ is a funnel if and only if for every nonisolated vertex $v$,
\begin{equation}\label{eq:funnel_path_counts}
\min\{\pi^{-}(v),\pi^{+}(v)\}=1.
\end{equation}
\end{lemma}

\begin{proof}
For any arc $e=(u,v)$, acyclicity gives $|\mathcal{P}_e|=\pi^{-}(u)\pi^{+}(v)$.
Indeed, every $s$-$u$ path combines with $e$ and every $v$-$t$ path to form a simple $s$-$t$ path through $e$.
Also, $\pi^{-}(v),\pi^{+}(v)\ge1$ for every nonisolated vertex by Assumption~\ref{ass:effective_network}.

For necessity, suppose that $\pi^{-}(w)\ge2$ and $\pi^{+}(w)\ge2$ for some vertex $w$.
Choose an $s$-$t$ path $P$ through $w$.
Varying the $w$-$t$ suffix shows that every arc of $P[s,w]$ is shared.
Varying the $s$-$w$ prefix shows that every arc of $P[w,t]$ is shared.
Acyclicity ensures that these replacements give simple paths.
Thus $P$ has no private arc, so $G$ is not a funnel.

Conversely, assume~\eqref{eq:funnel_path_counts} and fix $P\in\mathcal{P}$.
If $\pi^{-}(t)=1$, then $P$ is the only $s$-$t$ path and all its arcs are private.
Otherwise, let $v$ be the first vertex of $P$ with $\pi^{-}(v)\ge2$, and let $e=(u,v)$ be the arc of $P$ entering $v$.
By the choice of $v$, we have $\pi^{-}(u)=1$.
Condition~\eqref{eq:funnel_path_counts} gives $\pi^{+}(v)=1$, so $|\mathcal{P}_e|=\pi^{-}(u)\pi^{+}(v)=1$.
Hence $e$ is private, proving that $G$ is a funnel.
\end{proof}

The criterion says that every nonisolated vertex has either a unique path from $s$ or a unique path to $t$.
In particular, every nonisolated vertex other than $s$ and $t$ has indegree one or outdegree one.
This local degree condition is necessary but not sufficient.
Under the same assumptions, $G$ is a funnel if and only if no vertex of indegree at least two can reach a vertex of outdegree at least two, possibly itself.
Indeed, $\pi^{-}(v)\ge2$ holds exactly when a vertex of indegree at least two can reach $v$.
Similarly, $\pi^{+}(v)\ge2$ holds exactly when $v$ can reach a vertex of outdegree at least two.
Thus, once paths have merged, they cannot branch again downstream.

We can now relate convexity to network structure, capacity sufficiency, pathwise additivity, and Harsanyi dividends.

\begin{theorem}\label{thm:convexity_characterization}
Let $\Gamma_D=(E,\gamma)$ be a flow game on a network $D=(V,E;c;s,t)$ satisfying Assumption~\ref{ass:effective_network}.
The following statements are equivalent:
\begin{enumerate}[\emph{(}i\emph{)}]
    \item\label{cond:convexity} \textbf{Convexity}: The flow game $\Gamma_D$ is convex. 
    \item\label{cond:network_structure} \textbf{Network Structure}: The underlying topology $G_D$ is a funnel, and
    \begin{equation}\label{eq:shared_arc_capacity}
    c(e)\ge\sum_{P\in\mathcal{P}_e}c(P)
    \qquad\text{for every shared arc }e.
    \end{equation}
    \item\label{cond:capacity_sufficiency} \textbf{Capacity Sufficiency}: 
    The network $D$ satisfies capacity sufficiency:
    \[c(e)\ge\sum_{P\in\mathcal{P}_e}c(P) \qquad\text{for every arc }e.\]
    \item\label{cond:pathwise_additivity} \textbf{Pathwise Additivity}: The maximum-flow value equals the sum of all path capacities:
    \[\gamma(E)=\sum_{P\in\mathcal{P}}c(P).\]
    \item\label{cond:harsanyi_dividends} \textbf{Harsanyi Dividends}: The characteristic function $\gamma$ admits a non-negative path-based dividend decomposition:
    \[ \gamma = \sum_{P \in \mathcal{P}} c(P) \zeta_P. \]
\end{enumerate}
\end{theorem}

\begin{proof}
We prove the five statements equivalent through the following chain:
\[
(\ref{cond:convexity})\implies(\ref{cond:network_structure})\implies(\ref{cond:capacity_sufficiency})\implies(\ref{cond:pathwise_additivity})\implies(\ref{cond:harsanyi_dividends})\implies(\ref{cond:convexity}).
\]

(\emph{\ref{cond:convexity}}) $\implies$ (\emph{\ref{cond:network_structure}}):
Proposition~\ref{prop:acyclicity} gives acyclicity.
By Proposition~\ref{prop:bottleneck_exclusivity}, every path has a bottleneck belonging to no other path.
Thus $G_D$ is a funnel.
Proposition~\ref{prop:capacity_sufficiency} gives~\eqref{eq:shared_arc_capacity} for every shared arc.

(\emph{\ref{cond:network_structure}}) $\implies$ (\emph{\ref{cond:capacity_sufficiency}}):
At a private arc $e$, the unique path $P$ containing it satisfies $c(P)\le c(e)$.
At a shared arc $e$,~\eqref{eq:shared_arc_capacity} gives $\sum_{P\in\mathcal{P}_e}c(P)\le c(e)$.
Thus all paths can simultaneously carry their full capacities.

(\emph{\ref{cond:capacity_sufficiency}}) $\implies$ (\emph{\ref{cond:pathwise_additivity}}):
The path-flow vector $z_P=c(P)$ is feasible for~\eqref{eq:path_flow_primal} with $S=E$.
Every feasible path-flow vector satisfies $z_P\le c(P)$ for each path, so $\gamma(E)=\sum_{P\in\mathcal{P}}c(P)$.

(\emph{\ref{cond:pathwise_additivity}}) $\implies$ (\emph{\ref{cond:harsanyi_dividends}}):
Choose an optimal solution $\boldsymbol{z}^*$ of the path-flow problem~\eqref{eq:path_flow_primal} with $S=E$.
For each path $P$, feasibility gives $z^*_P\le c(e)$ for every $e\in P$, and hence $z^*_P\le c(P)$.
Since
\[
\sum_{P\in\mathcal{P}}z^*_P=\gamma(E)=\sum_{P\in\mathcal{P}}c(P),
\]
the non-negative differences $c(P)-z^*_P$ sum to zero.
Thus $z^*_P=c(P)$ for every $P\in\mathcal{P}$.

Now fix any coalition $S\subseteq E$.
Restricting this feasible path-flow vector to $\mathcal{P}_S$ gives a feasible flow in $D_S$, so
\[
\gamma(S)\ge\sum_{P\in\mathcal{P}_S}c(P).
\]
Conversely, every feasible solution of~\eqref{eq:path_flow_primal} on $D_S$ assigns at most $c(P)$ units to each path $P\in\mathcal{P}_S$.
Therefore
\[
\gamma(S)=\sum_{P\in\mathcal{P}_S}c(P)
=\sum_{P\in\mathcal{P}}c(P)\zeta_P(S),
\]
which yields the path-based Harsanyi dividend decomposition (\emph{\ref{cond:harsanyi_dividends}}).
If $\mathcal{P}=\emptyset$, every coalition has value zero and the same conclusion holds with an empty sum.

(\emph{\ref{cond:harsanyi_dividends}}) $\implies$ (\emph{\ref{cond:convexity}}):
Every unanimity game $\zeta_P$ is convex. 
Since $c(P) > 0$ for all $P\in \mathcal{P}$, and any non-negative linear combination of convex games remains convex, it follows that $\Gamma_D$ is convex.
\end{proof}

Theorem~\ref{thm:convexity_characterization} shows that a topology satisfying Assumption~\ref{ass:effective_network} admits a convex flow game with strictly positive capacities only if it is a funnel.

\subsection{Polyhedral Characterizations}
\label{subsec:polyhedral_characterization}

Throughout this subsection, $D$ satisfies Assumption~\ref{ass:effective_network}.
We separate the roles of topology and capacities in the geometry of feasible flows.
Theorem~\ref{thm:primal_convexity_characterization} gives the primal characterization: convexity is equivalent to the path-flow polytope being a box.
For a convex game, the arc-flow polytope is linearly isomorphic to this box, and the path-flow box structure holds in every subnetwork.
Theorem~\ref{thm:dual_convexity_characterization} gives the dual characterization: convexity is equivalent to a minimum $s$-$t$ cut consisting entirely of private arcs.
We then describe the optimal dual polytope and construct optimal solutions for every coalition.
For a topology $G$ satisfying Assumption~\ref{ass:effective_network}, let $\mathcal{K}(G)$ be the arc-flow cone of non-negative arc flows satisfying flow conservation at every vertex other than $s$ and $t$, without capacity bounds.
This cone includes circulations.
Write $\boldsymbol{\chi}_R\in\{0,1\}^{E}$ for the incidence vector of a path or cycle $R$.
Flow decomposition expresses every element of $\mathcal{K}(G)$ as a sum of path flows and cycle flows.
If $G$ is acyclic, this gives
\begin{equation}\label{eq:flow_cone_path_representation}
\mathcal{K}(G)=\left\{\sum_{P\in\mathcal{P}}z_P\boldsymbol{\chi}_P:
\boldsymbol{z}\in\mathbb{R}_{\ge0}^{\mathcal{P}}\right\}.
\end{equation}

\begin{lemma}\label{lem:unique_path_decomposition}
Let $G=(V,E;s,t)$ be a directed topology satisfying Assumption~\ref{ass:effective_network}.
Then $G$ is a funnel if and only if every flow $f\in\mathcal{K}(G)$ has a unique path decomposition
\[
f=\sum_{P\in\mathcal{P}}z_P\boldsymbol{\chi}_P,
\qquad \boldsymbol{z}\in\mathbb{R}_{\ge0}^{\mathcal{P}}.
\]
\end{lemma}

\begin{proof}
If $G$ is a funnel, it is acyclic, so~\eqref{eq:flow_cone_path_representation} gives a path decomposition for every flow in $\mathcal{K}(G)$.
Choose a private arc $e_P$ on each path $P$.
For any decomposition of a flow $f$, the coordinate at $e_P$ gives $z_P=f(e_P)$, so all coefficients are uniquely determined.

Conversely, suppose that every flow in $\mathcal{K}(G)$ admits a unique path decomposition.
If $G$ contained a directed cycle $C$, then $\boldsymbol{\chi}_C\in\mathcal{K}(G)$ would be a nonzero flow of value zero.
Each path vector has flow value one, so a non-negative combination $\sum_{P}z_P\boldsymbol{\chi}_P$ has flow value $\sum_{P}z_P$, which is zero only if $z_P=0$ for every $P$.
Thus $\boldsymbol{\chi}_C$ would admit no path decomposition, a contradiction.
Hence $G$ is acyclic.

If $G$ is not a funnel, Lemma~\ref{lem:funnel_path_counts} now gives a vertex $v$ with $\pi^{-}(v)\ge2$ and $\pi^{+}(v)\ge2$.
Choose distinct $s$-$v$ paths $A_1,A_2$ and distinct $v$-$t$ paths $B_1,B_2$.
Acyclicity ensures that their concatenations $P_{ij}=A_iB_j$ are four distinct simple $s$-$t$ paths.
The identity
\[
\boldsymbol{\chi}_{P_{11}}+\boldsymbol{\chi}_{P_{22}}
=\boldsymbol{\chi}_{P_{12}}+\boldsymbol{\chi}_{P_{21}}
\]
gives two distinct path decompositions of the same flow, contradicting uniqueness.
Hence $G$ is a funnel.
\end{proof}

Thus the funnel condition alone guarantees uniqueness of the path coordinates.
The additional capacity condition determines when these coordinates can vary independently up to their path capacities.
The path-flow polytope is
\begin{equation}\label{eq:path_flow_polytope}
\mathcal{F}_{\mathcal{P}}(D)=
\left\{\boldsymbol{z}\in\mathbb{R}_{\ge0}^{\mathcal{P}}:
\sum_{P\in\mathcal{P}_e}z_P\le c(e)\ \forall e\in E\right\}.
\end{equation}
The arc-flow polytope of $D$ is $\mathcal{K}(G_D)\cap\prod_{e\in E}[0,c(e)]$.

\begin{theorem}[Primal characterization]\label{thm:primal_convexity_characterization}
Let $\Gamma_D=(E,\gamma)$ be a flow game on a network $D=(V,E;c;s,t)$ satisfying Assumption~\ref{ass:effective_network}.
Then $\Gamma_D$ is convex if and only if
\begin{equation}\label{eq:path_flow_box}
\mathcal{F}_{\mathcal{P}}(D)=\prod_{P\in\mathcal{P}}[0,c(P)].
\end{equation}
\end{theorem}

\begin{proof}
Every feasible path-flow vector satisfies $z_P\le c(P)$ for each $P\in\mathcal{P}$, by the capacity constraint at a bottleneck of $P$.
Hence $\mathcal{F}_{\mathcal{P}}(D)$ is contained in the box $\prod_{P\in\mathcal{P}}[0,c(P)]$.
Since lowering any path flow preserves feasibility, equality holds if and only if the vector $z_P=c(P)$ for every $P\in\mathcal{P}$ is feasible.
By Theorem~\ref{thm:convexity_characterization}(\emph{\ref{cond:capacity_sufficiency}}), this is equivalent to convexity.
\end{proof}

Suppose that $\Gamma_D$ is convex.
The topology $G_D$ is a funnel by Theorem~\ref{thm:convexity_characterization}, so Lemma~\ref{lem:unique_path_decomposition} gives every flow $f\in\mathcal{K}(G_D)$ a unique path decomposition.
For any choice of representative bottlenecks $b_P$, bottleneck exclusivity makes $b_P$ private to $P$, so the path coordinates are recovered as $z_P=f(b_P)$.
Together with~\eqref{eq:path_flow_box}, this makes the map $\boldsymbol{z}\mapsto\sum_{P\in\mathcal{P}}z_P\boldsymbol{\chi}_P$ a linear isomorphism between the path-flow box and the arc-flow polytope.
The capacity constraints therefore reduce to one upper bound per path:
\begin{equation}\label{eq:arc_flow_bottleneck_box}
\mathcal{K}(G_D)\cap\prod_{e\in E}[0,c(e)]
=\left\{f\in\mathcal{K}(G_D):
0\le f(b_P)\le c(P)\ \forall P\in\mathcal{P}\right\}.
\end{equation}

The flow value $\sum_{P\in\mathcal{P}}z_P$ is uniquely maximized on the box by $z_P=c(P)$ for every $P\in\mathcal{P}$.
The linear isomorphism therefore gives a unique maximum flow $f^*$, with
\begin{equation}\label{eq:unique_maximum_flow}
f^*(e)=\sum_{P\in\mathcal{P}_e}c(P)\qquad\text{for every }e\in E.
\end{equation}

The same box structure holds in every subnetwork.
For every coalition $S\subseteq E$ and every arc $e\in S$, capacity sufficiency gives
\[
\sum_{P\in\mathcal{P}_S\cap\mathcal{P}_e}c(P)\le\sum_{P\in\mathcal{P}_e}c(P)\le c(e),
\]
so the vector $z_P=c(P)$ for every $P\in\mathcal{P}_S$ is feasible in $D_S$.
The same argument as in the proof above yields $\mathcal{F}_{\mathcal{P}_S}(D_S)=\prod_{P\in\mathcal{P}_S}[0,c(P)]$.

For a convex game on an arbitrary input network, these flow descriptions apply to $\widehat{D}$ obtained by deleting null arcs as in Lemma~\ref{lem:null_arc_reduction}.
The maximum-flow uniqueness established above concerns $\widehat{D}$ itself; maximum flows in the original network may differ by circulations involving null arcs.

We next characterize convexity through optimal dual solutions.
For any network $D$ satisfying Assumption~\ref{ass:effective_network}, let $\mathcal{Y}^*(D)$ denote the set of optimal solutions of~\eqref{eq:path_flow_dual} for $S=E$.
Strong duality gives
\begin{equation}\label{eq:general_optimal_dual_polytope}
\mathcal{Y}^*(D)=
\left\{\boldsymbol{y}\in\mathbb{R}_{\ge0}^{E}:
\sum_{e\in P}y_e\ge1\ \forall P\in\mathcal{P},\quad
\sum_{e\in E}c(e)y_e=\gamma(E)\right\}.
\end{equation}
This set is a bounded face of the dual feasible polyhedron, hence a polytope.
Indeed, optimality and $c(e)>0$ give $0\le y_e\le\gamma(E)/c(e)$ for every arc.

\begin{samepage}
\begin{theorem}[Dual characterization]\label{thm:dual_convexity_characterization}
Let $\Gamma_D=(E,\gamma)$ be a flow game on a network $D=(V,E;c;s,t)$ satisfying Assumption~\ref{ass:effective_network}.   
Then $\Gamma_D$ is convex if and only if there exists $\boldsymbol{y}^*\in\mathcal{Y}^*(D)$ such that $y_e^*=0$ for every shared arc $e$.
Equivalently, $D$ has a minimum $s$-$t$ cut consisting entirely of private arcs.
\end{theorem}
\end{samepage}

\begin{proof}
A minimum $s$-$t$ cut consisting of private arcs has an incidence vector in $\mathcal{Y}^*(D)$ that vanishes on all shared arcs.
Proposition~\ref{prop:bottleneck_cuts} supplies such a cut whenever $\Gamma_D$ is convex.

Conversely, suppose that $\boldsymbol{y}^*\in\mathcal{Y}^*(D)$ vanishes on all shared arcs.
Every arc with a positive dual price belongs to exactly one path.
By optimality, the definition of path capacity, and dual feasibility,
\[
\begin{aligned}
\gamma(E)
&=\sum_{e\in E}c(e)y_e^*
=\sum_{P\in\mathcal{P}}\sum_{e\in P}c(e)y_e^*\\
&\ge\sum_{P\in\mathcal{P}}c(P)\sum_{e\in P}y_e^*
\ge\sum_{P\in\mathcal{P}}c(P).
\end{aligned}
\]
The reverse inequality follows from $z_P\le c(P)$ for every feasible path-flow vector.
Thus $\gamma(E)=\sum_{P\in\mathcal{P}}c(P)$, and Theorem~\ref{thm:convexity_characterization}(\emph{\ref{cond:pathwise_additivity}}) gives convexity.
This proves both equivalences.
\end{proof}

Suppose that $\Gamma_D$ is convex, and let $f^*$ be the unique maximum flow given by~\eqref{eq:unique_maximum_flow}.
The polytope of optimal dual solutions is
\begin{equation}\label{eq:optimal_dual_polytope}
\mathcal{Y}^*(D)=
\left\{\boldsymbol{y}\in\mathbb{R}_{\ge0}^{E}:
\sum_{e\in P}y_e=1\ \forall P\in\mathcal{P},\quad
y_e=0\text{ if }c(e)>f^*(e)\right\}.
\end{equation}
Indeed, for any dual-feasible vector $\boldsymbol{y}$, the identities $f^*(e)=\sum_{P\in\mathcal{P}_e}c(P)$ and $\gamma(E)=\sum_{P\in\mathcal{P}}c(P)$ give
\[
\sum_{e\in E}c(e)y_e-\gamma(E)
=\sum_{e\in E}\bigl(c(e)-f^*(e)\bigr)y_e
+\sum_{P\in\mathcal{P}}c(P)\left(\sum_{e\in P}y_e-1\right).
\]
Every term on the right is non-negative, and $c(P)>0$ for every path.
The gap is therefore zero exactly when both conditions in~\eqref{eq:optimal_dual_polytope} hold.
Conversely, any vector in the displayed set is dual-feasible and has zero gap, so it is optimal.
Other optimal dual solutions may assign positive prices to saturated shared arcs.

Optimal dual solutions can also be chosen consistently across coalitions.
Fix one representative bottleneck $b_P$ for each path $P\in\mathcal{P}$.
For each $S\subseteq E$ and each $e\in S$, set
\begin{equation}\label{eq:representative_dual_solution}
B^*_S=\{b_P:P\in\mathcal{P}_S\},
\qquad
y^*_{S,e}=\mathbf{1}_{\{e\in B^*_S\}}.
\end{equation}
Bottleneck exclusivity ensures that each path in $\mathcal{P}_S$ contains exactly one representative, so $\boldsymbol{y}^*_S$ is feasible for~\eqref{eq:path_flow_dual} and $B^*_S$ is an $s$-$t$ cut of $D_S$.
The path-unanimity decomposition in Theorem~\ref{thm:convexity_characterization} gives
\[
\sum_{e\in S}c(e)y^*_{S,e}
=\sum_{P\in\mathcal{P}_S}c(P)=\gamma(S),
\]
so this dual solution is optimal and this cut is minimum.
Since $\mathcal{P}_S\subseteq\mathcal{P}_T$ whenever $S\subseteq T$, these solutions are monotone: $y^*_{S,e}\le y^*_{T,e}$ for every $e\in S$.
The construction allows null arcs in $D_S$; if $\mathcal{P}_S=\emptyset$, it gives the zero dual solution.

\subsection{Recognizing Convex Flow Games}
\label{subsec:convexity_recognition}

For an arbitrary network $D$, deleting the null arcs does not affect convexity (Lemma~\ref{lem:null_arc_reduction}).
Thus $\Gamma_D$ is convex if and only if the topology $G_{\widehat{D}}$ is a funnel and the shared-arc inequalities~\eqref{eq:shared_arc_capacity} hold on $\widehat{D}$.
Acyclicity is required of $\widehat{D}$, not of the original network with its null arcs.
Algorithm~\ref{alg:convexity_recognition} performs any necessary reduction before testing the funnel and capacity conditions.
If $D$ is known to satisfy Assumption~\ref{ass:effective_network}, the reduction is skipped.
For the capacity test, a \emph{private segment} is a maximal directed segment consisting of private arcs.
In a funnel, these segments correspond bijectively to the $s$-$t$ paths, as shown in the proof below.

\begin{algorithm}[H]
\caption{Recognizing convex flow games}\label{alg:convexity_recognition}
\small
\begin{algorithmic}[1]
\Require A flow network $D=(V,E;c;s,t)$
\Ensure \textsc{Yes} if $\Gamma_D$ is convex, and \textsc{No} otherwise
\State $F\gets E$
\If{$D$ is not known to satisfy Assumption~\ref{ass:effective_network}}
    \If{$D$ is acyclic}
        \State Find $R_s$, the vertices reachable from $s$, and $R_t$, those that can reach $t$
        \State $F\gets\{(u,v)\in E:u\in R_s,\ v\in R_t\}$
    \Else
        \State Enumerate up to $|E|+1$ simple $s$-$t$ paths by the procedure below\label{algline:recognition_enumeration}
        \State \textbf{if} $|E|+1$ paths are found \textbf{then} \Return \textsc{No}
        \State Set $F$ to the union of all listed paths
    \EndIf
\EndIf
\State Work on $\widehat{D}=(V,F;c|_F;s,t)$
\State \textbf{if} $F=\emptyset$ \textbf{then} \Return \textsc{Yes}
\State \textbf{if} $\widehat{D}$ contains a cycle \textbf{then} \Return \textsc{No}
\State Compute $\pi^{-}$ and $\pi^{+}$ in forward and reverse topological order, recording any count greater than two as two
\State \textbf{if} $\min\{\pi^{-}(v),\pi^{+}(v)\}\ne1$ for some nonisolated $v$ \textbf{then} \Return \textsc{No}
\State Identify private arcs $(u,v)$ by $\pi^{-}(u)=\pi^{+}(v)=1$ and form their private segments
\State For each private arc $e$, set $\ell(e)$ to the minimum capacity of its segment
\For{each $v\ne s$ with $\pi^{+}(v)>1$, in reverse topological order}
    \State For the unique incoming arc $e$ of $v$, set $\ell(e)$ to the sum of $\ell(a)$ over all outgoing arcs $a$ of $v$
\EndFor
\For{each $v\ne t$ with $\pi^{-}(v)>1$, in forward topological order}
    \State For the unique outgoing arc $e$ of $v$, set $\ell(e)$ to the sum of $\ell(a)$ over all incoming arcs $a$ of $v$
\EndFor
\State \Return \textsc{Yes} if $\ell(e)\le c(e)$ for every $e\in F$, and \textsc{No} otherwise
\end{algorithmic}
\end{algorithm}

The path enumeration in line~\ref{algline:recognition_enumeration} uses depth-first search on simple prefixes, starting with the length-zero prefix at $s$.
For a prefix ending at $u\ne t$ with vertex set $W$, compute the vertices that can reach $t$ after deleting $W$.
Extend the prefix only along an arc $(u,v)$ whose head belongs to this set.
When a prefix reaches $t$, output it and stop the search if $|E|+1$ paths have been output.
Parallel arcs are distinguished by their identities.

\begin{theorem}\label{thm:convexity_recognition}
Algorithm~\ref{alg:convexity_recognition} decides whether the flow game $\Gamma_D$ on any network $D=(V,E;c;s,t)$ is convex in time $O(|V||E|^2)$.
If $D$ is acyclic or is known to satisfy Assumption~\ref{ass:effective_network}, then $O(|V|+|E|)$ time suffices.
\end{theorem}

\begin{proof}
We first prove correctness and then bound the running time.
If $D$ is known to satisfy Assumption~\ref{ass:effective_network}, then $\widehat{E}=E$ and no arcs need to be deleted.
If $D$ is acyclic, an arc $(u,v)$ lies on an $s$-$t$ path exactly when $u$ is reachable from $s$ and $t$ is reachable from $v$.
Acyclicity ensures that these two paths concatenate with $(u,v)$ to form a simple path.
Thus the reachability test identifies exactly the arcs in $\widehat{E}$.

For the remaining networks, suppose first that $\Gamma_D$ is convex.
By Lemma~\ref{lem:null_arc_reduction}, $\Gamma_{\widehat{D}}$ is convex.
Theorem~\ref{thm:convexity_characterization}(\emph{\ref{cond:network_structure}}) therefore implies that $G_{\widehat{D}}$ is a funnel.
Choosing one private arc from each path gives $|\mathcal{P}|\le|\widehat{E}|\le|E|$.
Therefore, finding $|E|+1$ paths proves that the game is not convex.
At each prefix, the reachability test allows exactly those extensions that can be completed to a simple $s$-$t$ path.
Thus no path is listed twice, and every explored prefix other than the initial one has a completion.
If enumeration finishes with at most $|E|$ paths, all paths have been found and their union is $\widehat{E}$.

We may now work on $\widehat{D}$ by Lemma~\ref{lem:null_arc_reduction}.
If $\widehat{E}=\emptyset$, the game is zero and hence convex.
Otherwise, Theorem~\ref{thm:convexity_characterization}(\emph{\ref{cond:network_structure}}) rules out convexity when $G_{\widehat{D}}$ is cyclic or is not a funnel.
By Lemma~\ref{lem:funnel_path_counts}, recording two in place of any larger path count preserves both the funnel test and the identification of private arcs.

It remains to justify the capacity test on a funnel.
Along each path, the actual path counts $\pi^{-}$ are nondecreasing and $\pi^{+}$ are nonincreasing.
Since every path contains a private arc, its private arcs form one nonempty consecutive segment.
Every interior vertex of such a segment has $\pi^{-}(v)=\pi^{+}(v)=1$, and therefore has indegree and outdegree one.
Thus the private segments correspond bijectively to the $s$-$t$ paths.

For a shared arc $e=(u,v)$, either $\pi^{+}(v)>1$ or $\pi^{-}(u)>1$, and the funnel condition prevents both.
In the first case, $\pi^{-}(v)=1$, so $e$ is the unique incoming arc of $v$ and the first loop computes $\ell(e)$.
In the second, $\pi^{+}(u)=1$, so $e$ is the unique outgoing arc of $u$ and the second loop computes $\ell(e)$.
Thus $\ell(e)$ is computed exactly once for each shared arc $e$.
In either loop, the chosen topological order ensures that all values in the sum have already been computed.
The paths through the arc are partitioned according to the incident arcs in this sum.
Induction therefore shows that $\ell(e)$ is the sum, over all paths containing $e$, of the minimum capacities of their private segments.

If the final capacity inequalities hold, the minimum capacity of the private segment of each path $P$ is at most $\ell(e)\le c(e)$ for every $e\in P$.
This minimum is attained on $P$, so it equals $c(P)$.
Therefore, for any $e\in \widehat{E}$,
\begin{equation}\label{eq:capacity_test}
\ell(e)=\sum_{P\in\mathcal{P}_e}c(P)\le c(e).
\end{equation}
Theorem~\ref{thm:convexity_characterization}(\emph{\ref{cond:capacity_sufficiency}}) gives convexity of $\Gamma_{\widehat{D}}$.
Conversely, if $\Gamma_{\widehat{D}}$ is convex, Proposition~\ref{prop:bottleneck_exclusivity} implies that every path $P$ has a bottleneck in its private segment.
The minimum capacity of that segment is therefore $c(P)$, and Proposition~\ref{prop:capacity_sufficiency} gives $\ell(e)\le c(e)$ for every arc.
Together with Lemma~\ref{lem:null_arc_reduction}, this proves correctness for the original network.

We now bound the running time.
Checking acyclicity and computing the path counts take $O(|V|+|E|)$ time.
For an acyclic network, finding $\widehat{E}$ by reachability has the same bound.
The private segments have disjoint arc sets, so finding them and their minimum capacities also takes $O(|V|+|E|)$ time.
Each of the two loops uses each arc in at most one sum, and the final test compares $\ell(e)$ with $c(e)$ once for every arc $e$.
Thus the capacity test takes $O(|V|+|E|)$ time.
If $D$ is acyclic or is known to satisfy Assumption~\ref{ass:effective_network}, no path enumeration is needed, proving the linear-time bound.

For a network requiring path enumeration, let $q\le|E|+1$ be the number of paths listed.
Every explored prefix other than the initial one extends to a listed path.
If the search from such a prefix finishes, at least one path extending it has been listed.
If enumeration stops after finding $|E|+1$ paths, all prefixes still being examined lie on the last listed path.
Since each simple path has at most $\min\{|V|-1,|E|\}$ arcs, at most $1+q\min\{|V|-1,|E|\}$ prefixes are examined.
For each prefix, finding reachable vertices, examining outgoing arcs, and maintaining the current path take $O(|V|+|E|)$ time.
Since $|E|+1\le2|E|$ and
\[
\min\{|V|-1,|E|\}(|V|+|E|)\le2|V||E|,
\]
enumeration takes $O(|V||E|^2)$ time.
If all paths are listed, forming their union takes $O(|V||E|)$ time.
The remaining steps take $O(|V|+|E|)$ time, so the total time is $O(|V||E|^2)$.
For rational capacities, each segment's minimum is an input capacity, and each value $\ell(e)$ is a sum of at most $|E|$ such numbers.
Thus all intermediate values have polynomial encoding length, and the algorithm runs in polynomial time in the input encoding length.
\end{proof}

The enumeration step in line~\ref{algline:recognition_enumeration} can be accelerated using Johnson's algorithm~\cite{J1975}.
Construct an auxiliary graph by omitting isolated nonterminal vertices, subdividing each arc once, and adding a new vertex $r$ with arcs $(r,s)$ and $(t,r)$.
The resulting simple directed graph has $O(|E|)$ vertices and arcs, and its simple cycles through $r$ correspond bijectively to the original $s$-$t$ paths.
With $r$ ordered first, run only the phase rooted at $r$ in its strongly connected component.
The time before the first output, between consecutive outputs, and after the last is $O(|E|)$~\cite[Lemma~3]{J1975}.
Stopping after $|E|+1$ outputs or at exhaustion therefore takes $O(|E|^2)$ time.
Including preprocessing, formation of $F$, and the remaining steps, this replacement gives total time $O(|V|+|E|^2)$.
The two linear-time cases in Theorem~\ref{thm:convexity_recognition} remain unchanged.

If the algorithm returns \textsc{Yes}, all $s$-$t$ paths can be recovered by extending each private segment along its unique prefix from $s$ and suffix to $t$.
There are at most $|E|$ segments, and each resulting path has at most $|V|-1$ arcs.
Thus all paths can be recovered in $O(|V||E|)$ additional time.

\subsection{Path-Based Profit Sharing in Convex Flow Games}
\label{subsec:path_based_profit_sharing}

In this subsection, $D$ satisfies Assumption~\ref{ass:effective_network}, and $\Gamma_D$ is convex.
For an arbitrary input network, apply the constructions to $\widehat{D}$ and give each deleted null player zero payoff.
For every convex game, each core allocation can be extended to a PMAS~\cite{Spru90}.
For flow games, the path-unanimity decomposition $\gamma = \sum_{P\in\mathcal{P}}c(P)\zeta_P$ gives an explicit construction with a compact representation.

By a result of Shapley~\cite{Shap71}, this decomposition gives the core as a Minkowski sum:
\begin{equation}\label{eq:core_minkowski_sum}
\mathcal{C}(\Gamma_D) = \sum_{P \in \mathcal{P}} c(P) \mathcal{C}(\zeta_P).
\end{equation}
Each path core is the simplex $\mathcal{C}(\zeta_P) = \mathrm{conv}\{ \mathbf{1}_e : e \in P \}$, where $\mathbf{1}_e\in\mathbb{R}^{E}$ is the standard basis vector for arc $e$.

Given $\boldsymbol{x}\in\mathcal{C}(\Gamma_D)$, write $\boldsymbol{x}=\sum_{P\in\mathcal{P}}\boldsymbol{x}_P$ with $\boldsymbol{x}_P\in c(P)\mathcal{C}(\zeta_P)$.
Each $\boldsymbol{x}_P$ vanishes outside $P$ and has non-negative coordinates summing to $c(P)$.
For $e\in P$, write $x_{P,e}=\alpha_{P,e}c(P)$, where $\sum_{e\in P}\alpha_{P,e}=1$ and $\alpha_{P,e}\ge0$.
Keep these weights fixed across coalitions.
For every nonempty coalition $S\subseteq E$ and each $e\in S$, define
\begin{equation}\label{eq:path_based_pmas}
x_{S,e} = \sum_{P \in \mathcal{P}_e \cap \mathcal{P}_S} \alpha_{P,e} c(P).
\end{equation}
For $S=E$, this recovers the given allocation $\boldsymbol{x}$.

The scheme is efficient by the normalization property $\sum_{e \in P} \alpha_{P,e} = 1$:
\[
\sum_{e \in S} x_{S,e} = \sum_{e \in S} \sum_{P \in \mathcal{P}_e \cap \mathcal{P}_S} \alpha_{P,e} c(P) = \sum_{P \in \mathcal{P}_S} c(P) \sum_{e \in P} \alpha_{P,e} = \gamma(S).
\]
If $S\subseteq T$, then $\mathcal{P}_S\subseteq\mathcal{P}_T$.
Since all dividends and weights are non-negative, $x_{S,e}\le x_{T,e}$ for every $e\in S$, proving population monotonicity.
We can also assign all dividends to bottleneck arcs by setting $\alpha_{P,b_P} = 1$.
This assigns each path's dividend to the representative bottleneck selected by the optimal dual solution in~\eqref{eq:representative_dual_solution}.

Direct evaluation of~\eqref{eq:shapley_value_definition} involves $2^{|E|-1}$ coalitions for each player.
Linearity of the Shapley value and equal division among the arcs of each path-unanimity game give
\begin{equation}\label{eq:path_based_shapley_value}
\Phi_e(\gamma) = \sum_{P \in \mathcal{P}_e} \frac{c(P)}{|P|}.
\end{equation}
This corresponds to setting $\alpha_{P,e} = 1/|P|$ for all $e \in P$ in the PMAS construction~\eqref{eq:path_based_pmas}.
As shown in the proof of Theorem~\ref{thm:convexity_recognition}, a convex instance has at most $|E|$ simple $s$-$t$ paths, each with at most $|V|-1$ arcs.
Once these paths have been listed, computing all Shapley payoffs takes $O(|V||E|)$ additional time.
The PMAS requires storing $O(|V||E|)$ weights, and the allocation for any specified coalition can be evaluated in $O(|V||E|)$ time.
For rational capacities, the equal path weights $\alpha_{P,e}=1/|P|$ give a representation with polynomial encoding length.

\section{Conclusion}
\label{sec:conclusion}

Our characterization explains convexity through the independence of path contributions.
In a convex flow game, an arc's marginal contribution is the sum of the capacities of the paths it completes.
A larger coalition can only enlarge the set of such paths.
The capacity conditions make these contributions simultaneously attainable.
The box structure of the path-flow polytope and the monotone construction of optimal dual solutions from private bottlenecks express this independence in primal and dual terms.

Even after deleting null arcs, alternative routes and shared arcs need not prevent convexity.
Bottleneck exclusivity gives each path a private limiting resource, while capacity sufficiency ensures that shared arcs support the combined use of these resources.

Path independence allows each path's value to be divided using non-negative weights that remain fixed across coalitions.
These divisions yield a PMAS, with equal division giving the Shapley value.

A natural direction is to determine which forms of capacity competition remain compatible with population monotonicity.
A PMAS can exist without convexity, so independent path contributions are sufficient but not necessary for consistent profit sharing.
Studying how departures from the funnel and capacity conditions affect PMAS existence would clarify the boundary between increasing marginal contributions and population monotonicity.
A related direction is to extend the convexity characterization to networks in which each player may control several arcs.
Since coalitions then make groups of arcs available together, the question is how ownership interacts with topology and capacities to determine convexity.


\bibliographystyle{habbrv}
\bibliography{reference}

@incollection{Hars59,
  author    = {Harsanyi, John C.},
  title     = {A bargaining model for the cooperative $n$-person game},
  booktitle = {Contributions to the Theory of Games, Volume IV},
  editor    = {Luce, Robert D. and Tucker, Albert W.},
  pages     = {325--356},
  publisher = {Princeton University Press},
  year      = {1959}
}

@book{Schr03,
  title     = {{Combinatorial Optimization} - {Polyhedra} and {Efficiency}},
  publisher = {Springer},
  author    = {Schrijver, Alexander},
  year      = {2003}
}

@article{DIN99,
  author  = {Deng, Xiaotie and Ibaraki, Toshihide and Nagamochi, Hiroshi},
  title   = {Algorithmic Aspects of the Core of Combinatorial Optimization Games},
  journal = {Mathematics of Operations Research},
  volume  = {24},
  number  = {3},
  pages   = {751-766},
  year    = {1999}
}

@article{FKK01,
  title     = {On the computation of the nucleolus of a cooperative game},
  author    = {Faigle, Ulrich and Kern, Walter and Kuipers, Jeroen},
  journal   = {International Journal of Game Theory},
  volume    = {30},
  number    = {1},
  pages     = {79--98},
  year      = {2001},
  publisher = {Springer}
}

@article{MPS72,
  title     = {The kernel and bargaining set for convex games},
  author    = {Maschler, Michael and Peleg, Bezalel and Shapley, Lloyd S.},
  journal   = {International Journal of Game Theory},
  volume    = {1},
  number    = {1},
  pages     = {73--93},
  year      = {1972},
  publisher = {Springer}
}

@article{Spru90,
  title     = {Population monotonic allocation schemes for cooperative games with transferable utility},
  author    = {Sprumont, Yves},
  journal   = {Games and Economic Behavior},
  volume    = {2},
  number    = {4},
  pages     = {378--394},
  year      = {1990},
  publisher = {Elsevier}
}

@article{FZCD02,
  title   = {On computational complexity of membership test in flow games and linear production games},
  volume  = {31},
  number  = {1},
  journal = {International Journal of Game Theory},
  author  = {Fang, Qizhi and Zhu, Shanfeng and Cai, Maocheng and Deng, Xiaotie},
  year    = {2002},
  pages   = {39--45}
}

@article{Shap71,
  title   = {Cores of convex games},
  volume  = {1},
  number  = {1},
  journal = {International Journal of Game Theory},
  author  = {Shapley, Lloyd S.},
  year    = {1971},
  pages   = {11--26}
}

@article{MPS79,
  title   = {Geometric properties of the kernel, nucleolus, and related solution concepts},
  volume  = {4},
  number  = {4},
  journal = {Mathematics of Operations Research},
  author  = {Maschler, M. and Peleg, B. and Shapley, L. S.},
  year    = {1979},
  pages   = {303--338}
}

@article{PRB06,
  title   = {The nucleolus of balanced simple flow networks},
  volume  = {54},
  number  = {1},
  journal = {Games and Economic Behavior},
  author  = {Potters, Jos and Reijnierse, Hans and Biswas, Amit},
  year    = {2006},
  pages   = {205--225}
}

@article{GG92,
  title   = {On some network flow games},
  volume  = {17},
  number  = {4},
  journal = {Mathematics of Operations Research},
  author  = {Granot, Daniel and Granot, Frieda},
  year    = {1992},
  pages   = {792--841}
}

@article{KP09,
  title   = {On the core and $f$-nucleolus of flow games},
  journal = {Mathematics of Operations Research},
  volume  = {34},
  number  = {4},
  author  = {Kern, Walter and Paulusma, Dani\"{e}l},
  year    = {2009},
  pages   = {981--991}
}

@inproceedings{DFS06,
  title   = {Finding nucleolus of flow game},
  booktitle = {Proceedings of the Seventeenth Annual {ACM-SIAM} Symposium on Discrete Algorithms},
  author  = {Deng, Xiaotie and Fang, Qizhi and Sun, Xiaoxun},
  year    = {2006},
  pages   = {124--131}
}

@article{RMPT96,
  title   = {Simple flow games},
  volume  = {16},
  number  = {2},
  journal = {Games and Economic Behavior},
  author  = {Reijnierse, Hans and Maschler, Michael and Potters, Jos and Tijs, Stef},
  year    = {1996},
  pages   = {238--260}
}

@article{KZ82a,
  title   = {Generalized network problems yielding totally balanced games},
  volume  = {30},
  number  = {5},
  journal = {Operations Research},
  author  = {Kalai, Ehud and Zemel, Eitan},
  year    = {1982},
  pages   = {998--1008}
}

@article{KZ82b,
  title   = {Totally balanced games and games of flow},
  volume  = {7},
  number  = {3},
  journal = {Mathematics of Operations Research},
  author  = {Kalai, Ehud and Zemel, Eitan},
  year    = {1982},
  pages   = {476--478}
}

@article{GV85,
  title     = {Substitutes, complements and ripples in network flows},
  author    = {Granot, Frieda and Veinott, Jr, Arthur F.},
  journal   = {Mathematics of Operations Research},
  volume    = {10},
  number    = {3},
  pages     = {471--497},
  year      = {1985},
  publisher = {INFORMS}
}

@incollection{Shap53,
  author    = {Shapley, Lloyd S.},
  title     = {A value for $n$-person games},
  booktitle = {Contributions to the Theory of Games, Volume II},
  editor    = {Kuhn, Harold W. and Tucker, Albert W.},
  pages     = {307--317},
  publisher = {Princeton University Press},
  year      = {1953}
}

@article{GKV19,
  title   = {Multi-player flow games},
  volume  = {33},
  number  = {6},
  journal = {Autonomous Agents and Multi-Agent Systems},
  author  = {Guha, Shibashis and Kupferman, Orna and Vardi, Gal},
  year    = {2019},
  pages   = {798--820}
}

@article{BHS25,
  title   = {Flow allocation games},
  volume  = {50},
  number  = {1},
  journal = {Mathematics of Operations Research},
  author  = {Bertschinger, Nils and Hoefer, Martin and Schmand, Daniel},
  year    = {2025},
  pages   = {68--89}
}

@article{CCCW21,
  title   = {Atomic dynamic flow games: {Adaptive} vs. {Nonadaptive} agents},
  volume  = {69},
  number  = {6},
  journal = {Operations Research},
  author  = {Cao, Zhigang and Chen, Bo and Chen, Xujin and Wang, Changjun},
  year    = {2021},
  pages   = {1680--1695}
}

@article{J1975,
  author  = {Johnson, Donald B.},
  title   = {Finding all the elementary circuits of a directed graph},
  journal = {SIAM Journal on Computing},
  volume  = {4},
  number  = {1},
  pages   = {77--84},
  year    = {1975},
  doi     = {10.1137/0204007}
}

\end{document}